\documentclass[conference]{IEEEtran}
\makeatletter
\def\ps@headings{%
\def\@oddhead{\mbox{}\scriptsize\rightmark \hfil \thepage}%
\def\@evenhead{\scriptsize\thepage \hfil \leftmark\mbox{}}%
\def\@oddfoot{}%
\def\@evenfoot{}}
\makeatother
\usepackage{amsthm}
\usepackage{float}
\usepackage{graphicx,subfigure}
\usepackage{psfrag}
\usepackage{algorithm,algorithmic}
\usepackage{array}
\usepackage{amsmath,amssymb}
\usepackage{wasysym}
\usepackage{paralist}
\usepackage{color}
\usepackage{subfigure}
\usepackage{float}
\usepackage[font={small}]{caption}

\newtheorem{theorem}{Theorem}

\begin{document}

%\title{On Budget Constrained Region-based Fault Tolerant Design of Distributed File Storage in Networks
%%Design of Distributed Data Storage Networks Robust Against Region-Based Faults
%\thanks{This research is supported in part by a grant from the U.S. Defense Threat Reduction Agency under grant number HDTRA1-09-1-0032 and by a grant from the U.S. Air Force Office of Scientific Research under grant number FA9550-09-1-0120.}
%}

\title{Identification of $\cal K$ Most Vulnerable Nodes in Multi-layered Network Using a New Model of Interdependency
%%Design of Distributed Dable Nata Storage Networks Robust Against Region-Based Faults
\thanks{This research is supported in part by a grant from the U.S. Defense Threat Reduction Agency under grant number HDTRA1-09-1-0032 and by a grant from the U.S. Air Force Office of Scientific Research under grant number FA9550-09-1-0120.}
}
\author{\IEEEauthorblockN{Arunabha Sen, Anisha Mazumder, Joydeep Banerjee, Arun Das and Randy Compton}
%\author{\IEEEauthorblockN{Arunabha (Arun) Sen}
\IEEEauthorblockA{ Computer Science and Engineering Program\\
\small School of Computing, Informatics and Decision System Engineering\\
\small Arizona State University\\
\small Tempe, Arizona 85287\\
\small Email: \{asen, amazumde, Joydeep.Banerjee, adas22, Randy.Compton\}@asu.edu}
}

\maketitle

\begin{abstract}
The critical infrastructures of the nation including the power grid and the communication network are highly interdependent.
Recognizing the need for a deeper understanding of the interdependency in a multi-layered network, significant efforts have been made by the research community in the last few years to achieve this goal. Accordingly a number of models have been proposed and analyzed. Unfortunately, most of the models are over simplified and, as such, they fail to capture the complex interdependency that exists between entities of the power grid and the communication networks involving a combination of conjunctive and disjunctive relations. To overcome the limitations of existing models, we propose a new model that is able to capture such complex interdependency relations. Utilizing this model, we provide techniques to identify the $\cal K$ most vulnerable nodes of an interdependent network. We show that the problem can be solved in polynomial time in some special cases, whereas for some others, the problem is NP-complete. We establish that this problem is equivalent to computation of a {\em fixed point} of a multilayered network system and we provide a technique for its computation utilizing Integer Linear Programming. Finally, we evaluate the efficacy of our technique using real data collected from the power grid and the communication network that span the Maricopa County of Arizona.
\end{abstract}

\section{Introduction }
In the last few years there has been an increasing awareness
in the research community that the critical infrastructures of
the nation are closely coupled in the sense that the well being
of one infrastructure depends heavily on the well being of another. A case in point is the interdependency between the electric power grid and the communication network. The power
grid entities, such as the SCADA systems that control power
stations and sub-stations, receive their commands through
communication networks, while the entities of communication
network, such as routers and base stations, cannot operate
without electric power. Cascading failures in the power grid,
are even more complex now because of the coupling between
power grid and communication network. Due to this coupling,
not only entities in power networks, such as generators and
transmission lines, can trigger power failure, communication
network entities, such as routers and optical fiber lines, can
also trigger failure in power grid. Thus it is essential that
the interdependency between different types of networks be
understood well, so that preventive measures can be taken to avoid cascading catastrophic failures in multi-layered network
environments.

Recognizing the need for a deeper understanding of the
interdependency in a multi-layered network, significant efforts
have been made in the research community in the last few
years to achieve this goal \cite{Bul10, Gao11, Sha11, Ros08, Zha05, Par13, Ngu13, Zus11}. Accordingly a number of models have been proposed and
analyzed. Unfortunately, many of the proposed models are
overly simplistic in nature and as such they fail to capture
the complex interdependency that exists between power grid
and communication networks. In a highly cited paper \cite{Bul10},  the
authors assume that every node in one network depends on one
and only one node of the other network. However, in a follow
up paper \cite{Gao11},  the same authors argue that this assumption may
not be valid in the real world and a single node in one network
may depend on more than one node in the other network. A
node in one network may be functional (``alive'') as long as
one supporting node on the other network is functional.  
%However, in a follow up study in \cite{Gao11}, the authors from the same research group argue that {\em ``in the real world, this assumption may
%not be valid. A single node in network A may depend on more than one node in network B and will function as long as at least one of its support nodes in network B is still functional. Similarly, a node in network B may depend on more than one support nodes in network A. As long as at least one of its support nodes functions, the node in network B will also function.''} x

%While in a follow up study in \cite{Gao11}, the authors argue that such an assumption is too simplistic and as such a single node in a network may depend on more than one node in the other network for proper functioning, such that survival of certain nodes in network B is essential for proper functioning of certain nodes in network A and vice versa.

Although this generalization can account for {\em disjunctive dependency} of a node in the $A$ network (say $a_i$) on more than one node in the $B$ network (say, $b_j$ and $b_k$), implying that $a_i$ may be ``alive'' as long as either $b_i$ or $b_j$ is alive, it cannot account for {\em conjunctive dependency} of the form when {\em both $b_j$ and $b_k$} has to be alive in order for $a_i$ to be alive. In a real network the dependency is likely to be even more complex involving both disjunctive and conjunctive components. For example, $a_i$ may be alive if (i) $b_j$ {\em and} $b_k$ {\em and} $b_l$ are alive, {\em or} (ii) $b_m$ {\em and} $b_n$ are alive, {\em or} (iii) $b_p$ is alive. The graph based interdependency models proposed in the literature \cite{Sha11, Ros08, Zha05, Cas13, Par13, Ngu13} including \cite{Bul10, Gao11} cannot capture such complex interdependency between entities of multilayer networks. In order to capture such complex interdependency, we propose a new model using Boolean logic.
Utilizing this comprehensive model, we provide techniques to identify the $\cal K$ most vulnerable nodes of an interdependent multilayered network system. We show that the this problem can be solved in polynomial time for some special cases, whereas for some others, the problem is NP-complete. We also show that this problem is equivalent to computation of a {\em fixed point} \cite{Fud2010} and we provide a technique utilizing Integer Linear Programming to compute that fixed point. Finally, we evaluate the efficacy of our technique using real data collected from power grid and communication networks that span Maricopa County of Arizona. 

\section{Interdependency Model}

We describe the model for an interdependent network with two layers. However, the concept can easily be generalized to deal with networks with more layers. Suppose that the network entities in layer 1 are referred to as the $A$ type entities, $A = \{a_1, \ldots, a_n\}$ and entities in layer 2 are referred to as the $B$ type entities, $B = \{b_1, \ldots, b_m\}$. If the layer 1 entity $a_i$ is {\em operational} if (i) the layer 2 entities $b_j, b_k, b_l$ are operational, or (ii) $b_m, b_n$ are operational, or (iii) $b_p$ is operational, we express it in terms of {\em live equations} of the form $a_i \leftarrow b_jb_kb_l + b_mb_n + b_p$.  The live equation for a $B$ type entity $b_r$ can be expressed in a similar fashion in terms of $A$ type entities. If $b_r$ is {\em operational} if (i) the layer 1 entities $a_s, a_t, a_u, a_v$ are operational, or (ii) $a_w, a_z$ are operational, we express it in terms of {\em live equations} of the form $b_r \leftarrow a_sa_ta_ua_v + a_wa_z$. It may be noted that the {\em live equations} only provide a {\em necessary condition} for entities such as $a_i$ or $b_r$ to be {\em operational}. In other words, $a_i$ or $b_r$ may {\em fail} independently and may be {\em not operational} even when the conditions given by the corresponding {\em live equations} are {\em satisfied}.
A live equation in general will have the following form: \( x_i \leftarrow \sum_{j = 1}^{T_i}\prod_{k = 1}^{t_j} y_{j, k}\) where $x_i$ and $y_{j, k}$ are elements of the set $A$ ($B$) and $B$ ($A$) respectively, $T_i$ represents the {\em number} of min-terms in the live equation and $t_j$ refers to the {\em size} of the $j$-th min-term (the size of a min-term is equal to the number of $A$ or $B$ elements in that min-term). In the example $a_i \leftarrow b_jb_kb_l + b_mb_n + b_p$, $T_i = 3$, $t_1= 3, t_2 = 2, t_3 = 1$, $x_i = a_i$, $y_{2, 1} = b_m$, $y_{2, 2} = b_p$.

We refer to the live equations of the form $a_i \leftarrow b_jb_kb_l + b_mb_n + b_p$ also as {\em First Order Dependency Relations}, because
these relations express direct dependency of the $A$ type entities
on $B$ type entities and vice-versa. It may be noted however
that as $A$ type entities are dependent on $B$ type entities,
which in turn depends on $A$ type entities, the failure of
some $A$ type entities can trigger the failure of other $A$ type
entities, though {\em indirectly}, through some $B$ type entities. Such
interdependency creates a {\em cascade of failures} in multilayered
networks when only a few entities of either $A$ type or $B$ type
(or a combination) fails. We illustrate this with the help of
an example. The live equations for this example is shown in
table~\ref{eqTbl}.

\begin{table}[H]
\begin{center}
\begin{tabular}{|l||l|}  \hline
{\bf Power Network} & {\bf Communication Network} \\ \hline
$a_1\leftarrow b_1 + b_2$ & $b_1 \leftarrow a_1 + a_2 a_3$ \\ \hline
$a_2 \leftarrow  b_1b_3 + b_2$ & $b_2 \leftarrow a_1 + a_3$ \\ \hline
$a_3 \leftarrow b_1b_2b_3$ & $b_3 \leftarrow a_1a_2$ \\ \hline
$a_4 \leftarrow b_1 + b_2 + b_3$ & $--$ \\ \hline
\end{tabular}
%\vspace{0.02in}
\caption{Live equations for a Multilayer Network}
\protect\label{eqTbl}
\end{center}
\end{table}

\begin{table}[H]
\begin{center}
\begin{tabular}{|c|c|c|c|c|c|c|c|}  \hline
\multicolumn{1}{|c|}{Entities} & \multicolumn{7}{c|}{Time Steps}\\
\cline{2-8} & $t_0$ & $t_1$ & $t_2$ & $t_3$ & $t_4$ & $t_5$ & $t_6$ \\\hline \hline
$a_1$ & $1$ & $1$ & $1$ & $1$ & $1$ & $1$ & $1$ \\ \hline
$a_2$ & $0$ & $0$ & $0$ & $0$ & $1$ & $1$ & $1$ \\ \hline
$a_3$ & $0$ & $0$ & $1$ & $1$ & $1$ & $1$ & $1$ \\ \hline
$a_4$ & $0$ & $0$ & $0$ & $0$ & $1$ & $1$ & $1$ \\ \hline
$b_1$ & $0$ & $0$ & $0$ & $1$ & $1$ & $1$ & $1$ \\ \hline
$b_2$ & $0$ & $0$ & $0$ & $1$ & $1$ & $1$ & $1$ \\ \hline
$b_3$ & $0$ & $1$ & $1$ & $1$ & $1$ & $1$ & $1$ \\ \hline
%{\em Optimal} & $a_0$ & $a_1$ &  $\cdots$ & $\cdots$ &  $a_k$ \\ \hline
%{\em Greedy } & $b_0$ & $b_1$ & $\cdots$ & $\cdots$ &  $b_k$ \\ \hline
\end{tabular}
%\vspace{0.02in}
\caption{Time Stepped Cascade Effect for a Multilayer Network}
\protect\label{cascadeTbl}
\end{center}
\end{table}

%\begin{figure*}[ht]
  %\subfigure[Cascading failures reach steady state after $p$ time steps]{\label{fig:cascade}\includegraphics[width=0.5\textwidth, keepaspectratio]{figure/cascadeEffect.pdf}}
  %\subfigure[Cascading failures as a fixed point system]{\label{fig:fixedPoint}\includegraphics[width=0.5\textwidth, keepaspectratio]{fixedPoint.eps}}
  %\caption{Cascading Failures in Multi-layered Networks}  \label{fig:cascadeFixedPoint}
%\end{figure*}
\vspace{-0.1in}

\begin{figure}[ht]
\begin{center}
 {\includegraphics[width=0.4\textwidth, keepaspectratio]{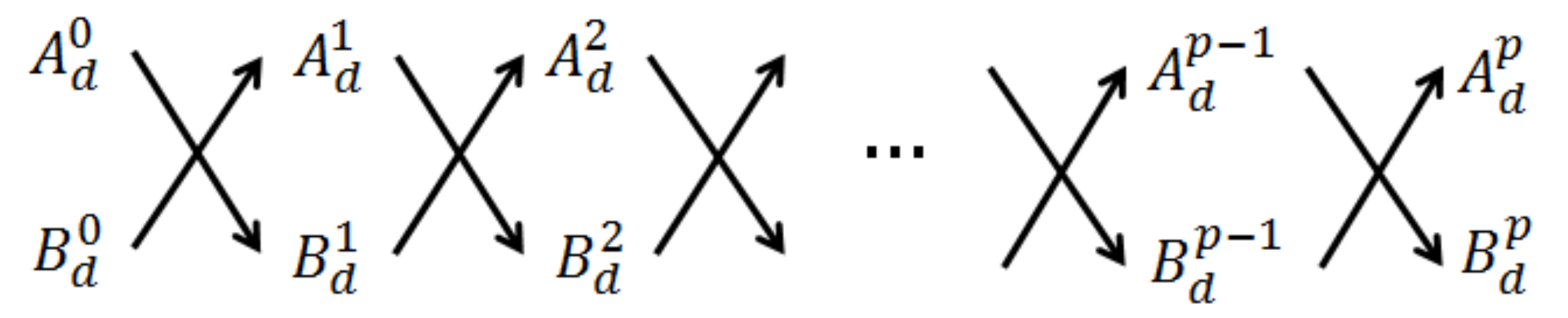}}
  \caption{Cascading failures reach steady state after $p$ time steps}
  %\caption{Plot of budget v/s the number of dead entities in the multi-layer graph composed of power network and communication network.}
\label{fig:cascade}
\end{center}
\end{figure}

As shown in table \ref{cascadeTbl}, the failure of only one entity $a_1$ at time step $t_0$ triggered a chain of failures that resulted in the failure of all the entities of the network after by timestep $t_4$. A table entry of 1 indicates that the entity is ``dead''. In this example, the failure of $a_1$ at $t_0$ triggered the failure of $b_3$ at $t_1$, which in turn triggered the failure of $a_3$ at $t_2$. The failure of $b_3$ at $t_1$ was due to the dependency relation $b_3 \leftarrow a_1a_2$ and the failure of $a_3$ at $t_2$ was due to the dependency relation $a_3 \leftarrow b_1b_2b_3$. The cascading failure process initiated by failure (or death) of a subset of $A$ type entities at timestep $t_0$, $A_{d}^0$ and a subset of $B$ type entities $B_{d}^{0}$ till it reaches its final steady state is shown diagrammatically in figure \ref{fig:cascade}. Accordingly, a multilayered network can be viewed as a ``closed loop'' control system. Finding the steady state after an initial failure in this case is equivalent of computing the {\em fixed point} of a function ${\mathcal F}(.)$ such that ${\mathcal F} (A_{d}^p \cup B_{d}^p) = A_{d}^p \cup B_{d}^p$, where $p$ represents the number of steps when the system reaches the steady state.

We define a set of ${\cal K}$ entities in a multi-layered network as most vulnerable, if failure of these ${\cal K}$ entities triggers the failure of the largest number of other entities. The goal of the ${\cal K}$ most vulnerable nodes problem is to identify this set of nodes. This is equivalent to identifying $A_{d}^0 \subseteq A$, $B_{d}^0 \subseteq B$, that {\em maximizes} $|A_{d}^p \cup B_{d}^p|$, subject to the constraint that $|A_{d}^0 \cup B_{d}^0| \le {\cal K}$.

The dependency relations (live equations) can be formed either after careful analysis of the multilayer network along the lines carried out in \cite{Zus11}, or after consultation with the engineers of the local utility and internet service providers.

%\begin{table}[H]
%\begin{center}
%\begin{tabular}{|c|c|c|}  \hline
%{\bf Case} & {\bf No. of Min-terms} & {\bf Size of Min-terms}\\ \hline
%Case $1$ & $1$ & $1$ \\ \hline
%Case $2$ & Arbitrary & $1$ \\ \hline
%Case $3$ & $1$ & Arbitrary \\ \hline
%Case $4$ & Arbitrary & Arbitrary \\ \hline
%\end{tabular}
%\caption{Equation Types for Dependency Relations}
%\protect\label{dependencyTbl}
%\end{center}
%\end{table}

\section{Computational Complexity and Algorithms}
Based on the number and the size of the min-terms in the dependency relations, we divide them into four different cases as shown in Table \ref{dependencyTbl}. The algorithms for finding the $\cal K$ most vulnerable nodes in the multilayer networks and computation complexity for each of the cases are discussed in the following four subsections.

\begin{table}[H]
\begin{center}
\begin{tabular}{|c|c|c|}  \hline
{\bf Case} & {\bf No. of Min-terms} & {\bf Size of Min-terms}\\ \hline
Case I & $1$ & $1$ \\ \hline
Case II & $1$ & Arbitrary \\ \hline
Case III & Arbitrary & $1$ \\ \hline
Case IV & Arbitrary & Arbitrary \\ \hline
\end{tabular}
\caption{Equation Types for Dependency Relations}
\protect\label{dependencyTbl}
\end{center}
\end{table}

\subsection{Case I: Problem Instance with One Min-term of Size One}

In this case, a live equation in general will have the following form: \( x_i \leftarrow y_{j}\) where $x_i$ and $y_{j}$ are elements of the set $A$ ($B$) and $B$ ($A$) respectively. In the example $a_i \leftarrow b_j$, $x_i = a_i$, $y_{1} = b_j$. It may be noted that a conjunctive implication of the form $a_i \leftarrow b_jb_k$ can also be written as two separate implications $a_i \leftarrow b_j$ and $a_i \leftarrow b_k$. However, such cases are considered in Case II and is excluded from consideration in Case I. The exclusion of such implications implies that the entities that appear on the LHS of an implication in Case I are unique. This property enables us to develop a polynomial time algorithm for the solution of the $\cal K$ most vulnerable node problem for this case. We present the algorithm next.

\vspace{0.05in}
\noindent
{\tt Algorithm 1}

\vspace{0.05in}
\noindent
{\tt Input:} (i) A set $S$ of implications of the form of $y \leftarrow x$, where $x,y \in A \cup B$, (ii) An integer $\cal K$.

\vspace{0.05in}
\noindent
{\tt Output:} A set $V\rq{}$ where $\vert V\rq{} \vert = {\cal K}$ and $V\rq{} \subset A \cup B$ such that failure of entities in $V\rq{}$ at time step $t_0$ results in failure of the largest number of entities in $A \cup B$ when the steady state is reached.

\vspace{0.05in}
\noindent
{\tt Step 1.} We construct a directed graph $G=(V,E)$, where $V= A \cup B$. For each implication $y \leftarrow x$ in $S$, where $x,y \in A \cup B$, we introduce a directed edge $(x,y) \in E$. 

\vspace{0.05in}
\noindent
{\tt Step 2.} For each node $x_i \in V$, we construct a transitive closure set $C_{x_i}$ as follows: If there is a path from $x_i$ to some node $y_i \in V$ in $G$, then we include $y_i$ in $C_{x_i}$. It may be recalled that  $\vert A \vert + \vert B \vert =  n + m$. So, we get $n+m$ transitive closure sets $C_{x_i}, 1 \le i \le (n+m)$. We call each $x_i$ to be the {\em seed entity} for the transitive closure set $C_{x_i}$.

\vspace{0.05in}
\noindent
{\tt Step 3.} We  remove all the transitive closure sets which are proper subsets of some other transitive closure set.

\vspace{0.05in}
\noindent
{\tt Step 4.} Sort the {\em remaining transitive closure sets} $C_{x_i}$, where the rank of the closure sets is determined by the cardinality of the sets. The sets with a larger number of entities are ranked higher than the sets with a fewer number of entities.

\vspace{0.05in}
\noindent
{\tt Step 5.} Construct the set $V\rq{}$ by selecting the seed entities of the top $\cal K$ transitive closure sets.  If the number of {\em remaining transitive closure sets} is less than $\cal K$ (say, ${\cal K}'$), arbitrarily select the remaining entities.
%then the set $V\rq{}$ is constructed by including all the 
%${\cal K}'$ seed entities of the {\em remaining transitive closure sets} and then selecting ${\cal K} - {\cal K}'$ arbitrary seed entities. 

\vspace{0.05in}
\noindent
{\em{Time complexity of Algorithm 1:}}  Step 1 takes $O(n + m + \vert S \vert)$ time. Step 2 can be executed in $O((n+m)^3)$ time. Step 3 takes at most $O((n+m)^2)$ time. Step 4 sorts at most $\vert S \vert$ entries, a standard sorting algorithm takes $O(\vert S \vert~log~\vert S \vert)$ time. Selecting $\cal K$ entities in step 5 takes $O(\cal K)$ time. Since $\vert S \vert \leq n+m$, hence the overall time complexity is $O((n+m)^3)$

\vspace{0.05in}
\noindent
\begin{theorem}{} For each pair of transitive closure sets $C_{x_i}$ and $C_{x_j}$ produced in step 2 of algorithm 1, either $C_{x_i} \cap C_{x_j} = \emptyset$ or $C_{x_i} \cap C_{x_j} = C_{x_i}$ or $C_{x_i} \cap C_{x_j} = C_{x_j}$, where $x_i \neq x_j$.
\end{theorem}

\vspace{0.05in}
\noindent
{  \em{Proof:}} Consider, if possible, that there is a pair of transitive closure sets $C_{x_i}$ and $C_{x_j}$ produced in step 2 of algorithm 1, such that $C_{x_i} \cap C_{x_j} \neq \emptyset$ and $C_{x_i} \cap C_{x_j} \neq C_{x_i}$ and $C_{x_i} \cap C_{x_j} \neq C_{x_j}$. Let $x_k \in C_{x_i} \cap C_{x_j} $. This implies that there is a path from $x_i$ to $x_k$ ($path_1$) as well as there is a path from $x_j$ to $x_k$, ($path_2$). Since, $x_i \not = x_j$ and $C_{x_i} \cap C_{x_j} \not = C_{x_i}$ and $C_{x_i} \cap C_{x_j} = C_{x_j}$, there is some $x_l$ in the $path_1$ such that $x_l$ also belongs to $path_2$. W.l.o.g, let us consider that $x_l$ be the first node in $path_1$ such that $x_l$ also belongs to $path_2$. This implies that $x_l$ has in-degree greater than $1$. This in turn implies that there are two implications in the set of implications $S$ such that $x_l$ appears in the L.H.S of both. This is a contradiction because this violates a characteristic of the implications in Case I. Hence, our initial assumption was wrong and the theorem is proven.
\vspace{0.05in}
\noindent
\begin{theorem}{} Algorithm 1 gives an optimal solution for the problem of selecting $\cal K$ most vulnerable entities in a multi-layer network for case I dependencies.
\end{theorem}

\vspace{0.05in}
\noindent
{\em{Proof:}} Consider that the set $V\rq{}$ returned by the algorithm is not optimal and the optimal solution is $V_{OPT}$. Let us consider there is a entity $x_i \in A \cup B$ such that $x_i \in V_{OPT} \setminus V\rq{}$. Evidently, (i) $C_{x_i}$ was either deleted in step 3 or (ii) $\vert C_{x_i} \vert $ is less than the cardinalities of all the transitive closure sets with seed entities $x_j \in V\rq{}$, because our algorithm did not select $x_i$. Hence, in both cases, replacing any entity $x_j \in V\rq{}$ by $x_i$ reduces the total number of entities killed. Thus, the number of dead entities by the failure of entities in $ V_{OPT}$ is lesser than that caused by the failure of the entities in $V\rq{}$, contradicting the optimality of $V_{OPT}$. Hence, the algorithm does in fact return the optimal solution.

\subsection{Case II: Problem Instance with One Min-term of Arbitrary Size }
In this case, a live equation in general will have the following form: \( x_i \leftarrow \prod_{k = 1}^{q} y_{j}\) where $x_i$ and $y_{j}$ are elements of the set $A$ ($B$) and $B$ ($A$) respectively, $q$ represents the {\em size} of min-term. In the example $a_i \leftarrow b_jb_kb_l$, $q = 3$, $x_i = a_i$, $y_{1} = b_j$, $y_{2} = b_k$, $y_{3} = b_k$.

\vspace{0.05in}
\noindent
\subsubsection{Computational Complexity}
We show that computation of $\cal K$ most vulnerable nodes ($\cal K$-MVN) in a multilayer network is NP-complete in Case II. We formally state the problem next.

\vspace{0.05in}
\noindent
{\tt Instance}: Given a set of dependency relations between $A$ and $B$ type entities in the form of live equations \( x_i \leftarrow \prod_{k = 1}^{q} y_{j}\), integers $\cal K$ and $\cal L$.\\
{\tt Question}: Is there a subset of $A$ and $B$ type entities of size at most $\cal K$ whose ``death'' (failure) at time $t_0$, triggers a cascade of failures resulting in failures of at least $\cal L$ entities, when the steady state is reached?
%\vspace{0.1in}
\vspace{0.05in}
\noindent
\begin{theorem}
The $\cal K$-MVN problem is NP-complete.
\end{theorem}

\vspace{0.05in}
\noindent
{\em{Proof:}} We prove that the $\cal K$-MVN problem is NP-complete by giving a transformation for the vertex cover (VC) problem.  An instance of the vertex cover problem is specified by an undirected graph $G = (V, E)$ and an integer $R$. We want to know if there is a subset of nodes $S \subseteq V$ of size at most $R$, so that every edge has at least one end point  in $S$. From an instance of the VC problem, we create an instance of the $\cal K$-MVN problem in the following way. First, from the graph $G = (V, E)$, we create a directed graph $G' = (V, E')$ by replacing each edge $e \in E$ by two oppositely directed edges $e_1$ and $e_2$ in $E'$ (the end vertices of $e_1$ and $e_2$ are same as the end vertices of $e$). Corresponding to a node $v_i$ in $G'$ that has incoming edges from other nodes (say) $v_j$, $v_k$ and $v_l$, we create a dependency relation (live equation) $v_i \leftarrow v_j v_k v_l$. We set $\cal K$ = $R$ and $\cal L$ = $|V|$. The corresponding death equation is of the form $\bar {v_i} \leftarrow \bar {v_j} + \bar{ v_k} + \bar {v_l}$ (obtained by taking {\em negation} of the live equation). We set $\cal K$ = $R$ and $\cal L$ = $|V|$. It can now easily be verified that if the graph $G = (V, E)$ has a vertex cover of size $R$ iff in the created instance of $\cal K$-MVN problem  death (failure) of at most $\cal K$ entities at time $t_0$, will trigger a cascade of failures resulting in failures of at least $\cal L$ entities, when the steady state is reached.

\subsubsection{Optimal Solution with Integer Linear Programming}

%If the live equation is in the {\em product} form \( x_i \leftarrow \prod_{k = 1}^{q} y_{j}\) then the ``death equation'' (obtained by taking {\em negation} of the live equation) will be in the {\em sum} form \(\bar {x_i} \leftarrow \sum_{j = 1}^{q} {\bar y_{j}}\). If the live equation is given as $a_i \leftarrow b_jb_kb_l$ then the death equation will be given as $\bar {a_i} \leftarrow \bar {b_j} + \bar {b_k} + \bar {b_l}$. 
In this case, we can find and optimal solution to the $\cal K$-MVN problem using Integer Linear Programming (ILP). We associate binary indicator variables $x_i$ ($y_i$) to capture the state of the entities $a_i$ ($b_i$). $x_i$ ($y_i$) is 1 when $a_i$ ($b_i$) is dead and 0 otherwise. Since we want find the set of $\cal K$ entities whose failure at time step $t_0$
triggers cascading failure resulting in the failure of the largest number of entities, the objective of the ILP can be written as follows
\({\tt maximize}~~~\sum_{i = 1}^n x_i + \sum_{i = 1}^m y_i\)
It may be noted that the variables in the objective function do not have any notion of {\em time}. However, cascading failure takes place in time steps, $a_i$ triggers failure of $b_j$ at time step $t_1$, which in turn triggers failure of $a_k$ in time step $t_2$ and so on. Accordingly, in order to capture the cascading failure process, we need  to introduce the notion of time into the variables of the ILP. If the numbers of $A$ and $B$ type entities are $n$ and $m$ respectively, the steady state must be reached by time step $n + m -1$ (cascading process starts at time step 0, $t_0$). Accordingly, we introduce $n + m$ versions of the variables $x_i$ and $y_i$, i.e., $x_i[0], \ldots, x_i[n + m -1]$ and $y_i[0], \ldots, y_i[n + m -1]$.
To indicate the state of entities $a_i$ and $b_i$ at times $t_0, \ldots, t_{n + m - 1}$. The objective of the ILP is now changed to
\[{\tt maximize} \sum_{i = 1}^n x_i[n + m -1] + \sum_{i = 1}^m y_i[n + m -1]\]

Subject to the constraint that no more than $\cal K$ entities can fail at time $t_0$.\\
\noindent
{\tt Constraint~1:} \(\sum_{i = 1}^n x_i[0] + \sum_{i = 1}^m y_i[0] \leq {\cal K}\)
\noindent
In order to ensure that the cascading failure process conforms to the dependency relations between type $A$ and $B$ entities, additional constraints must be imposed.\\
\noindent
{\tt Constraint 2}: If an entity fails at time fails at time step $p$, (i.e., $t_p$) it should continue to be in the failed state at all time steps $t >  p$. That is $x_i (t) \geq x_i (t-1), \forall t, 1 \leq t \leq n + m - 1$. Same constraint applies to $y_i (t)$.\\
\noindent
{\tt Constraint 3}: The dependency relation (death equation) $\bar {a_i} \leftarrow \bar {b_j} + \bar {b_k} + \bar {b_l}$ can be translated into a linear constraint in the following way \(x_i (t) \leq y_j (t -1) + y_k (t -1) + y_l(t -1), \forall t, 1 \leq t \leq n + m - 1\). 

\vspace{0.05in}
\noindent
The optimal solution to $\cal K$-MVN problem for Case II can be found by solving the above ILP.  

\subsection{Case III: Problem Instance with an Arbitrary Number of Min-terms of Size One}
%
%A live equation in this special case will have the following form: \[ a_i \leftarrow \sum_{j = 1}^{T_i} b_{j} ~~~{\tt or}~~~ b_i \leftarrow \sum_{j = 1}^{T_i} a_{j}\] where $T_i$ represents the {\em number} of min-terms in the live equation of entity $a_i$ or $b_i$.

A live equation in this special case will have the following form: \( x_i \leftarrow \sum_{j = 1}^{q} y_{j}\) where $x_i$ and $y_{j}$ are elements of the set $A$ ($B$) and $B$ ($A$) respectively, $q$ represents the {\em number} of min-terms in the live equation. In the example $a_i \leftarrow b_j + b_k + b_l$, $q = 3$, $x_i = a_i$, $y_{1} = b_j$, $y_{2} = b_k$, $y_{3} = b_l$.
\subsubsection{Computational Complexity}
We show that a special case of the problem instances with an arbitrary number of min-terms of size one is same as the {\em Subset Cover} problem (defined below), which is proven to be NP-complete. We define 
$Implication\_Set (A)$ to be the set of all implications of the form $a_i \leftarrow \sum_{j = 1}^{T_i} b_{j}$ and $Implication\_Set (B)$ to be the set of all implications of the form $b_i \leftarrow \sum_{j = 1}^{T_i} a_{j}$.
Now consider a subset of the set of problem instances with an arbitrary number of min-terms of size one where either $Implication\_Set (A) = \emptyset$ or $Implication\_Set (B) = \emptyset$.
Let \(A' = \{a_i | a_i ~\tt{is~the~element~on~the~LHS~of~an~implication}\}\)
\(\tt{in~the~Implication\_Set (A)}\). The set $B'$ is defined accordingly. If $Implication\_Set (B) = \emptyset$ then $B' = \emptyset$. In this case, failure of any $a_i, 1 \leq i \leq n$ type entities will not cause failure of any $b_j, 1 \leq j \leq m$ type entities. Since an adversary can cause failure of only $\cal K$ entities, the adversary would like to choose only those $\cal K$ entities that will cause failure of the largest number of entities. In this scenario, there is no reason for the adversary to attack any $a_i, 1 \leq i \leq n$ type entities as they will not cause failure of any $b_j, 1 \leq j \leq m$ type entities. On the other hand, if the adversary attacks  $\cal K$ $b_j$ type entities, not only those $\cal K$ $b_j$ type entities will be destroyed, some  $a_i$ type entities will also be destroyed due to the implications in the $Implication\_Set (A)$. As such the goal of the adversary will be to carefully choose $\cal K$ $b_j, 1 \leq j \leq m$ type entities that will destroy the largest number of $a_i$ type entities.  In its abstract form, the problem can be viewed as the Subset Cover problem. 

\noindent
{\tt Subset Cover Problem}\\
%\vspace{0.25in}
%\noindent
%{\bf Subset Cover Problem}\\
\noindent
{\em Instance}: A set $S = \{s_1, \ldots, s_m\}$, a set $\cal S$ of $m$ subsets of $S$, i.e., ${\cal S}$ = $\{S_1, \ldots, S_r\}$, where
$S_i \subseteq S, \forall i, 1 \leq i \leq r$, integers $p$ and $q$. \\
\noindent
{\em Question}: Is there a $p$ element subset $S'$ of $S$ ($p < n$) that completely covers at least $q$ elements of the set $\cal S$? (A set $S'$ is said to be {\em completely covering} an element $S_i, \forall i, 1 \leq i \leq m$ of the set $\cal S$, if $S' \cap S_i = S_i, \forall i, 1 \leq i \leq m$.)

The set $S$ in the subset cover problem corresponds to the set $B = \{b_1, \ldots, b_m\}$, and each set $S_i, 1 \leq i \leq r$ corresponds to an implication in the $Implication_Set(A)$ and comprises of
the $b_j$'s that appear on the RHS of the implication. The goal of the problem is to select a subset $B''$ of $B$ that maximizes the number of $S_i$'s completely covered by $B''$.
\newpage
\noindent
\begin{theorem}
The Subset Cover problem is NP-complete.
\end{theorem}
\begin{proof}
We prove that the Subset Cover problem is NP-complete by giving a transformation from the well known
Clique problem. It may be recalled that an instance of the
Clique problem is specified by a graph $G = (V, E)$ and an integer $K$. The decision question is whether or not a clique
of size at least $K$ exists in the graph $G = (V, E)$. We show
that a clique of size $K$ exists in graph $G = (V, E)$ iff the
Subset Cover problem instance has a $p$ element subset $S'$ of $S$ that completely covers at least $q$ elements of the set $\cal S$.

From an instance of the Clique problem, we create an instance of the Subset Cover problem in the following way. Corresponding to every vertex $v_i, 1 \leq i \leq n$ of the graph $G = (V, E)$ ($V = \{v_1, \ldots, v_n\}$), we create an element in the set $S = \{s_1, \ldots, s_n\}$. Corresponding to every edge $e_i, 1 \leq i \leq m$, we create $m$ subsets of $S$, i.e., ${\cal S}$ = $\{S_1, \ldots, S_m\}$, where $S_i$ corresponds to a two element subset of nodes, corresponding to the end vertices of the edge $e_i$. We set the parameters $p = K$ and $q = K(K - 1)/2$. Next we show that in the instance of the subset cover problem created by the above construction process, a $p$ element subset $S'$ of $S$ exists that completely covers at least $q$ elements of the set $\cal S$, iff the graph $G = (V, E)$ has a clique of size at least $K$.

Suppose that the graph $G = (V, E)$ has a clique of size $K$. It is clear that in the created instance of the subset cover problem, we will have $K(K - 1)/2$ elements in the set $\cal S$, that will be completely covered by a $K$ element subset of the set $S$. The $K$ element subset of $S$ corresponds to the set of $K$ nodes that make up the clique in $G = (V, E)$ and the $K(K - 1)/2$ elements in the set $\cal S$ corresponds to the edges of the graph $G = (V, E)$ that corresponds to the edges of the clique. Conversely, suppose that the instance of the Subset Cover problem has $K$ element subset of $S$ that completely covers $K(K - 1)/2$ elements of the set $\cal S$. Since the elements of $\cal S$ corresponds  to the edges in $G$, in order to completely cover $K(K-1)/2$ edges, at least $K$ nodes (elements of the set $S$) will be necessary. As such, this set of $K$ nodes will constitute a clique in the graph $G = (V, E)$.
\end{proof}
%\subsubsection{Inapproximibilty Result for Subset Cover Problem}
\subsubsection{Optimal Solution with Integer Linear Programming}

If the live equation is in the form \( x_i \leftarrow \sum_{k = 1}^{q} y_{j}\) then the ``death equation'' (obtained by taking {\em negation} of the live equation) will be in the {\em product} form \(\bar {x_i} \leftarrow \prod_{j = 1}^{q} {\bar y_{j}}\). If the live equation is given as 
${a_i} \leftarrow {b_j} + {b_k} $, then the death equation will be given as $\bar {a_i} \leftarrow \bar {b_j} \bar {b_k}$.

By associating binary indicator variables $x_i$ and $y_i$ to capture the state of the entities $a_i$ and $b_i$, we can follow almost identical procedure as in Case II, with only one exception. It may be recalled that in Case II, the death equations such as $\bar {a_i} \leftarrow \bar {b_j} + \bar {b_k}$ was translated into a linear constraint  \(x_i (t) \leq y_j (t -1) + y_k (t -1), \forall t, 1 \leq t \leq n + m - 1\). However a similar translation in Case III, with death equations such as $\bar {a_i} \leftarrow \bar {b_j} \bar {b_k}$, will result in a {\em non-linear} constraint of the form  \(x_i (t) \leq y_j (t -1) y_k (t -1), \forall t, 1 \leq t \leq n + m - 1\). Fortunately, a non-linear constraint of this form can be replaced a linear constraint such as \(2x_i (t) \leq y_j (t -1) + y_k (t -1), \forall t, 1 \leq t \leq n + m - 1\). After this transformation, we can compute the optimal solution using integer linear programming.

\subsection{Case IV: Problem Instance with an Arbitrary Number of Min-terms of Arbitrary Size}
%This is the most general version of the problem. 
\subsubsection{Computational Complexity}
Since both Case II and Case III are special cases of Case IV, the computational complexity of finding the $\cal K$ most vulnerable nodes in the multilayer network in NP-complete in Case IV also.
\subsubsection{Optimal Solution with Integer Linear Programming}
The optimal solution to this version of the problem can be computed by combining the techniques developed for the solution of the versions of the problems considered in Cases II and III.
%\subsubsection{Heuristic Solution (1/2 page)}

%\begin{table*}[t]
%\begin{center}
%\begin{tabular}{|c||c|c|c|c|c|c|c|}  \hline
%
%{\bf Region} & {\bf Constituent Zip Codes} & {\bf Generators} & {\bf Loads} &  {\bf Transmission Lines} & {\bf Cell-towers} & {\bf Fiberlit Buildings} & {\bf Fiberlinks} \\ \hline
%%0 & 0 & 0 & 0 &  0 &   $p$ \\ \hline
%$1$ & $85354$ & $5$ & $15$ & $24$ &   $10$  &  $5$ & $4$\\ \hline
%$2$ & $85072$ & $8$ & $12$ & $20$ &   $6$  &  $6$ & $6$\\ \hline
%$3$ & $85326,85338,85354$ & $5$ & $15$ & $24$ &   $6$  &  $9$ & $4$\\ \hline
%$4$ & $85307,85043,85009,85007$  & $5$ & $16$ & $28$ &   $8$  &  $8$ & $4$\\ \hline
%$5$ & $85301,85304,85308,85373,85387$ & $5$ & $15$ & $24$ & $5$ &   $10$  &  $5$ \\ \hline
%%{\em Optimal} & $a_0$ & $a_1$ &  $\cdots$ & $\cdots$ &  $a_k$ \\ \hline
%%{\em Greedy } & $b_0$ & $b_1$ & $\cdots$ & $\cdots$ &  $b_k$ \\ \hline
%\end{tabular}
%\vspace{0.02in}
%\caption{Regions of Maricopa county identified for failure vulnerability experiments}
%%\caption{$5$ regions of Maricopa county constituting our experimental dataset}
%\protect\label{expData}
%\end{center}
%\end{table*}

\begin{figure*}[ht]
\begin{center}
  \subfigure[Snapshot of Power Network in Maricopa County]{\includegraphics[width=0.45\textwidth, keepaspectratio]{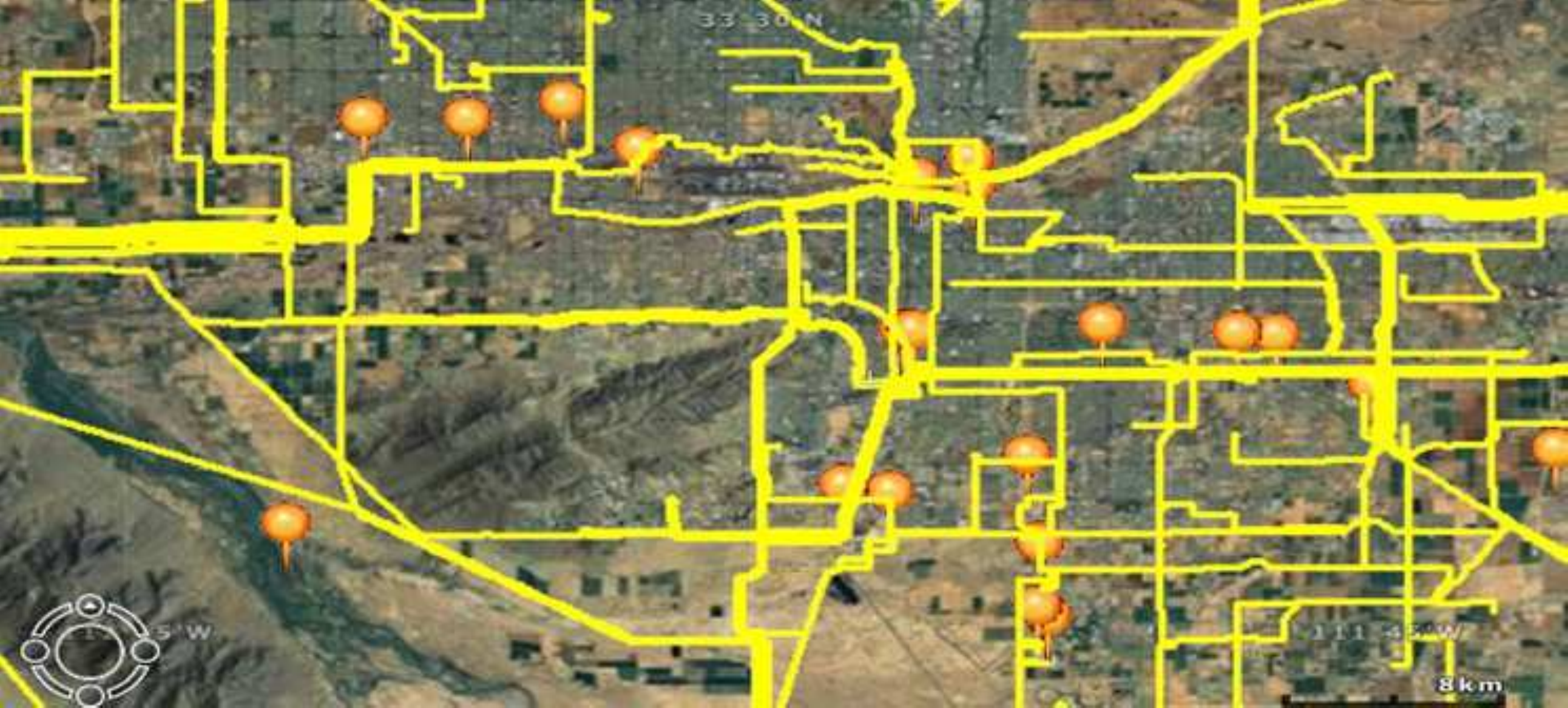}}
  \subfigure[Snapshot of Communication Network in Maricopa County]{\includegraphics[width=0.45\textwidth, keepaspectratio]{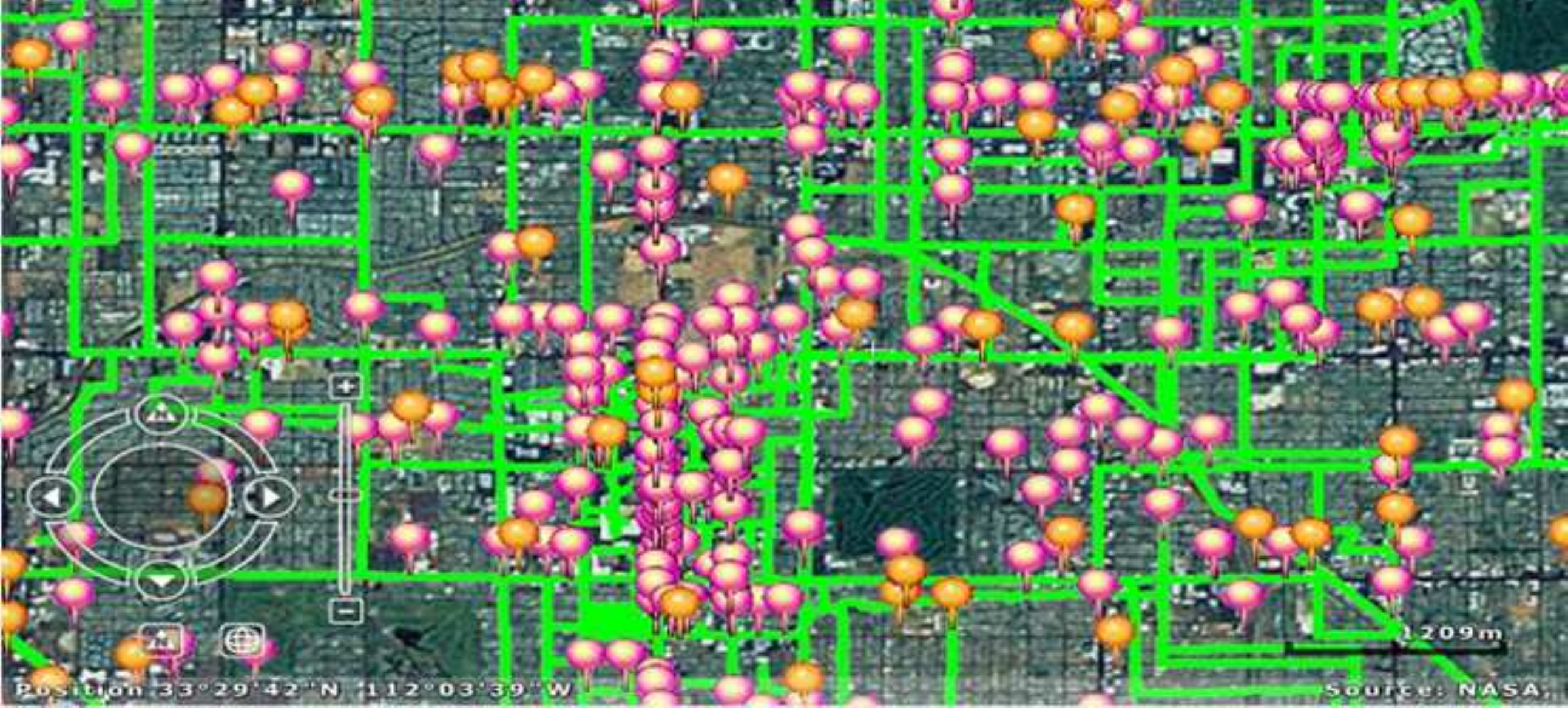}}
  \caption{Snapshots of power network and communication network in Maricopa County)}
\label{fig:powerCommNet}
\end{center}
\end{figure*}
\section{Experimental results}
We applied our model to study multilayer vulnerability
issues in Maricopa County, the most densely populated county
of Arizona with approximately 60\% of Arizona’s population
residing in it. Specifically, we wanted to find out if some
regions of Maricopa County were more vulnerable to failure
than some other regions. The data for our multi-layered
network were obtained from different sources. We obtained
the data for the power network (network A) from Platts (http://www.platts.com/). Our power network dataset consists of $70$ power plants and $470$ transmission lines. Our communication network (network B) data were obtained from GeoTel (http://www.geo-tel.com/). Our communication network data consists of $2,690$ cell towers and $7,100$ fiber-lit buildings as well as $42,723$ fiber links. Snapshots of our power network data and communication network data are shown in figure \ref{fig:powerCommNet}. In the power network snapshot of sub-figure(a), the orange markers show locations of powerplants while the yellow continuous lines represent the transmission lines. In the communication network snapshot of sub-figure (b) the pink markers show the location of fiber-lit buildings, the orange markers show the location of cell towers and the green continuous lines represent the fiber links. In our dataset, \lq{}load\rq{} in the Power Network is divided into Cell towers and Fiber-lit buildings. Although there exists various other physical entities which also draw electric power and hence can be viewed as load to the power network, as they are not relevant to our study on interdependency between power and communication networks, we ignore such entities. Thus in network A, we have the three types of Power Network Entities (PNE\rq{}s)  - Generators, Load (consisting of Cell towers and Fiber-lit buildings) and Transmission lines (denoted by $a_1,a_2,a_3$ respectively). For the Communication Network, we have the following Communication Network Entities (CNE\rq{}s) - Cell Towers, Fiber-lit buildings and Fiber links (denoted by $b_1,b_2,b_3$ respectively). We consider the Fiber-lit buildings as a communication network entities as they house routers which definitely are communication network entities. From the raw data we construct {\em Implication\_Set(A)} and {\em Implication\_Set(B)}, by following the rules stated below:

\vspace{0.05in}
\noindent
{\tt Rules:} We consider that a PNE is dependent on a set of CNEs for being in the active state (\lq{}alive\rq{}) or being in the inactive state (\lq{}dead\rq{}). Similarly, a CNE is dependent on a set of PNEs for being active or inactive state. For simplicity we consider the live equations with at most two minterms. For the same reason we consider the size of each minterm is at most two. 

%Thus, the \lq{}live\rq{} equations for the PNE\rq{}s and CNE\rq{}s are as follows:

\vspace{0.05in}
\noindent
%{\tt PNE\rq{}s}
{\em Generators ($a_{1, i}, 1 \leq i \leq p$, where $p$ is the total number of generators):}  We consider that each generator ($a_{1.i}$) is dependent on the nearest Cell Tower ($b_{1,j}$) or the nearest Fiber-lit building  ($b_{2,k}$) and the corresponding Fiber link ($b_{3,l}$) connecting $b_{2,k}$ and $a_{1,i}$. Thus, we have \\
$a_{1,i} \leftarrow b_{1,j}+b_{2,k} \times b_{3,l}$

\vspace{0.05in}
\noindent
{\em Load ($a_{2,i}, 1 \leq i \leq q$, where $q$ is the total number of loads):} We consider that the loads in the power network do not depend on any CNE.

\vspace{0.05in}
\noindent
{\em Transmission Lines ($a_{3,i}, 1 \leq i \leq r$, where $r$ is the total number of transmission lines):} We consider that the transmission lines do not depend on any CNE.

%{\tt CNE\rq{}s}
\vspace{0.05in}
\noindent
{\em Cell Towers ($b_{1,i}, 1 \leq i \leq s$, where $s$ is the total number of cell towers):} We consider the cell towers depend on the nearest pair of generators and the corresponding transmission line connecting the generator to the cell tower. Thus, we have $b_{1,i} \leftarrow a_{1,j} \times a_{3,k}+a_{1,j\rq{}} \times a_{3,k\rq{}}$

\vspace{0.05in}
\noindent
{\em Fiber-lit Buildings ($b_{2,i}, 1 \leq i \leq t$, where $t$ is the total number of fiber-lit buildings): } We consider that the fiber-lit buildings depend on the nearest pair of generators and the corresponding transmission lines connecting the generators to the cell tower. Thus, we have $b_{2,i} \leftarrow a_{1,j} \times a_{3,k}+a_{1,j\rq{}} \times a_{3,k\rq{}}$

\vspace{0.05in}
\noindent
{\em Fiber Links ($b_{3,i}, 1 \leq i \leq u$, where $u$ is the total number of fiber links)):} We consider that the fiber links do not depend on any PNE.
%We consider that the fiber links require power only for the amplifiers connected to them. Now, amplifiers are needed only if the length of the fiber link is higher than a certain threshold. In our dataset, we consider that only those fiber links which are 'quite long' need power. We further consider that these fiber links depend on the nearest pair of generators and the corresponding transmission lines connecting the generators with the fiber link under consideration. Thus, we have\\
%$b_{3,i} \leftarrow a_{1,j} \times a_{3,k}+a_{1,j\rq{}} \times a_{3,k\rq{}}$\\
%But because of lack of data about the length of our fiber links or the precise value of the threshold for the length of the fiber link beyond which it needs an amplifier, in our experiments we do not consider that the fiber links require any power.\\

\vspace{0.05in}
\noindent
Because of experimental resource limitation, we have considered $5$ regions of Maricopa County for our experiments.  
%Using the ILP formulation given in section III, we formulated the integer linear programs from the our experimental datasets. 
We used IBM CPLEX Optimizer 12.5 to run the formulated ILP's on the experimental dataset. We show our results in the figure \ref{fig:plot}. We observe that in each of the regions there is a specific
budget threshold beyond which each additional increment in
budget results in the death of only one entity. The reason
for this behavior is our assumption that entities such as the
transmission lines and the fiberlinks are not dependent on
any other entities. We notice that all the entities of the two
networks can be destroyed with a budget of about $60\%$ of the number of entities of the two networks $A$ and $B$. Most importantly, we find that the degree of vulnerability of all the five regions considered in our study are close and no one region stands out as being extremely vulnerable.  

\begin{figure}[ht]
\begin{center}
 {\includegraphics[width=0.4\textwidth, keepaspectratio]{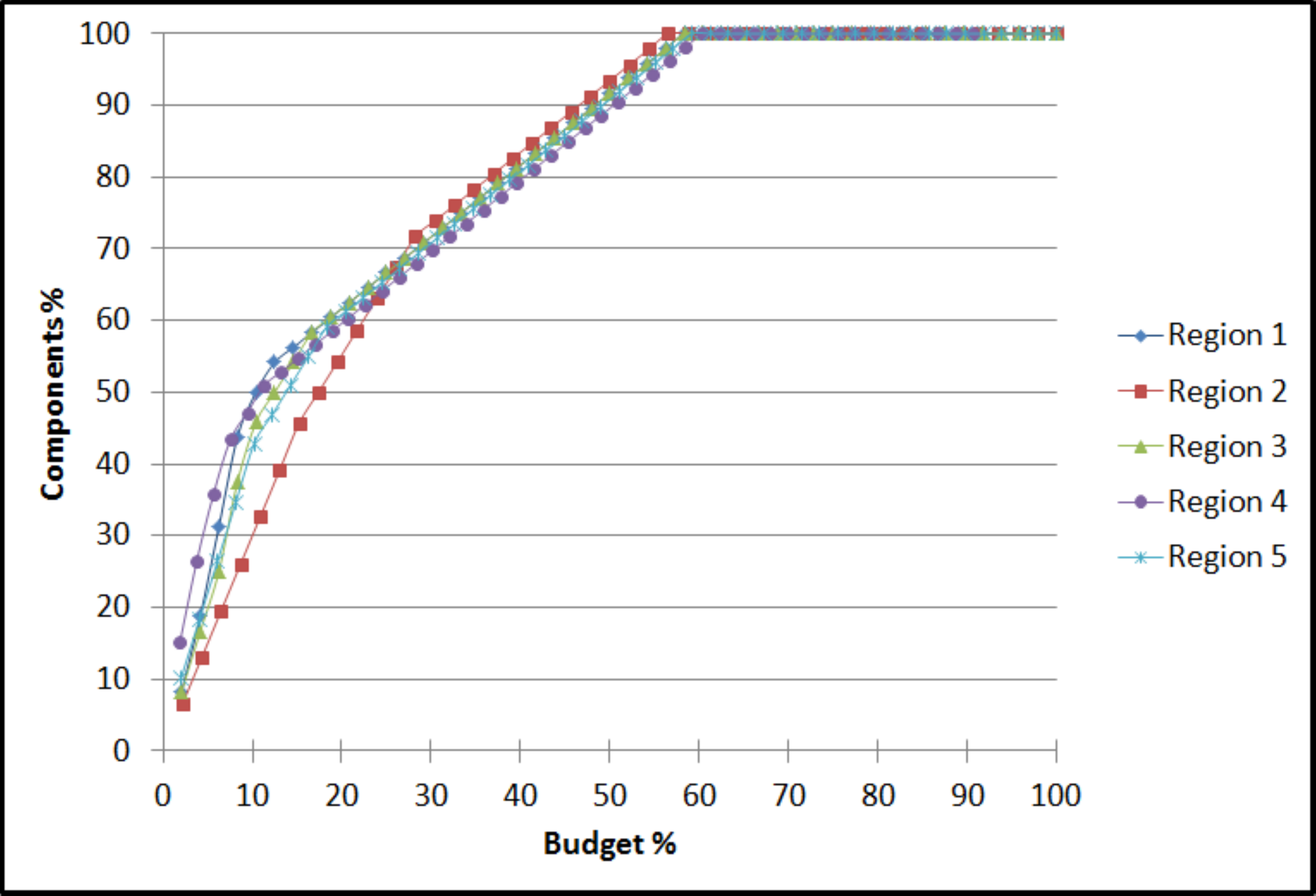}}
  \caption{Experimental results of failure vulnerability across five regions of Maricopa county}
  %\caption{Plot of budget v/s the number of dead entities in the multi-layer graph composed of power network and communication network.}
\label{fig:plot}
\end{center}
\end{figure}
\begin{footnotesize}
\bibliographystyle{IEEEtran}
\bibliography{IEEEabrv,references,referencesBibToAdd}
\end{footnotesize}
\end{document}